\documentclass[14pt, a4paper]{extarticle}
\usepackage[margin=2cm]{geometry}
\usepackage[affil-it]{authblk}

\usepackage[english]{babel}
\usepackage{graphicx} 
\usepackage{wrapfig}
\usepackage{amsmath,amsthm,amssymb,scrextend}
\usepackage{booktabs}            
\usepackage{multirow}            
\usepackage{amsfonts}            
\usepackage{multicol}
\usepackage[shortlabels]{enumitem}
\usepackage{cutwin}
\usepackage{float}
\usepackage{csquotes}
\usepackage{enumitem}
\usepackage{tikz}
\usepackage{hyperref}

\def \bf#1 {\textbf{#1 }}
\def \sumt {\sum\limits}
\providecommand{\keywords}[1]
{
  \small	
  \textbf{\textit{Keywords:}} #1
}

\renewenvironment{proof}{\begin{addmargin}[1em]{0em}\begin{newproof}}{\end{newproof}\end{addmargin}\qed}

\newtheorem{lm}{Lemma}
\newtheorem{cor}{Corollary}
\newtheorem*{lm*}{Lemma}
\newtheorem{defin}{Definition}
\newtheorem*{defin*}{Definition}

\def \sumt {\sum\limits}

\def \cN {\mathcal{N}}

\begin{document}

\title{New exponential law for real networks}
\author{Mikhail Tuzhilin%
  \thanks{Affiliation: Moscow State University, Electronic address: \texttt{mikhail.tuzhilin@math.msu.ru};}
}
\date{}
\maketitle

\begin{abstract}
In this article we have shown that the distributions of ksi satisfy an exponential law for real networks while the distributions of ksi for random networks are bell-shaped and closer to the normal distribution. The ksi distributions for Barabasi-Albert and Watts-Strogatz networks are similar to the ksi distributions for random networks (bell-shaped) for most parameters, but when these parameters become small enough, the Barabasi-Albert and Watts-Strogatz networks become more realistic with respect to the ksi distributions.
\end{abstract}



\section{Introduction}

For real networks (based on real data), there are two fundamental invariants --- a large average clustering coefficient~\cite{WS} and a power law in the degree distribution~\cite{BA1}. Both of these invariants are based on centrality: degree centrality and local clustering coefficient. These invariants give rise to two classes of networks: scale-free and small-world, respectively. It turns out that most real networks are both scale-free and small-world. In this article, we introduce a new invariant for real networks --- the exponential distribution law of centrality ksi~\cite{Tuz}. We call the class of networks with an exponential distribution law of centrality ksi exponential-ksi networks. To show that this is invariant, we took 40 real networks from 3 databases with different types, numbers of nodes and edges and showed that all of these networks are exponential-ksi networks. To see this more clearly, we present log plots similar to those made when the scale-free property was shown on a log-log scale~\cite{BA2}. Moreover, we showed that random networks have a bell-shaped distribution of ksi and thus are not exponential-ksi networks. We also tested these properties at Watts-Strogatz small-world~\cite{WS} and Barabasi-Albert scale-free~\cite{BA2} networks and showed that they are exponential-ksi networks only for a small interval of parameters. In this case, they are more tree-like and thus more similar to the shape of real networks. 

\section{Prerequisites}

Let's give the basic denotations. Consider connected undirected graph $G$ with $n$ vertices. Denote by $A = A(G) = \{a_{ij}\}$ adjacency matrix of $G$. Denote by $\cN(i)$ a neighborhood of the vertex $i$ (the vertices which are adjacent to $i$) and by $d_i$ the degree of $i$. For any two disjoint subsets of vertices $H, K\subset V(G)$ denote the number of edges with one end in $H$ and another in $K$ by $E(H, K) = \big|{(v,w): v\in H,\ w\in K}\big|$.

A ksi-centrality is defined by following:
\begin{defin}
    For each vertex $i$ \bf{ksi-centrality} $\xi_{i}$ is the relation of total number of neighbors of $i$-th neighbors except between themselves divided by the total number of neighbors of $i$:
    $$  
        \xi_i = \xi(i) = \frac {\Big|E\big(\cN(i), V\setminus \cN(i)\big)\Big|} {\big|\cN(i)\big|} = \frac {\Big|E\big(\cN(i), V\setminus \cN(i)\big)\Big|} {d_i}.
    $$
\end{defin}

\begin{defin}
    A network with an exponential distribution law of centrality ksi is called \bf{exponential-ksi network}.
\end{defin}

For quick computations, the value $\Big|E\big(\cN(i), V\setminus \cN(i)\big)\Big|$ can be found in terms of product of adjacency matrix by two columns  of adjacency matrix.

\begin{lm}
$$
    E\big(\cN(i), V\setminus \cN(i)\big) = \sumt_{j, k\in V(G)} a_{ij} a_{jk} \overline a_{ki},
$$ where $\overline a_{ki} = 1-a_{ki}.$
\end{lm}
\begin{proof}
    Let's fix $i$ and note that
    $$\sumt_{j\in V(G)} a_{ij} a_{jk} = \begin{cases}
        d_i, & k = i, \\
        1, & i\sim j\sim k, \\
        0, & \text{otherwise},
    \end{cases}
    \qquad\text{and}\qquad
    1 - a_{ki} = \begin{cases}
        1, & k = i, \\
        1, & k\not\sim i, \\
        0, & k\sim i.
    \end{cases}
    $$
    Therefore, 
    $$
        \Big|E\big(\cN(i), V\setminus \cN(i)\big)\Big| = d_i+ \big|k, j\in V(G): i\sim j\sim k, k\not\sim i\big| =  \sumt_{j, k\in V(G)} a_{ij} a_{jk} \overline a_{ki}.
    $$
\end{proof}

\begin{cor}\label{cor1} Let's $A$ be adjacency matrix of a graph for each vertex $i$
    $$
        \xi_i = \frac {\Big(A^2\cdot\overline A\Big)_{ii}} {\big(A^2\Big)_{ii}},
    $$
    where $\overline{A} = I-A$ for $I$ --- matrix of all ones.
\end{cor}

Since $\frac {\Big|E\big(\cN(i), V\setminus \cN(i)\big)\Big|} {d_i} = \frac {d_i} {d_i} = 1$, when vertices of $\cN(i)\cup \{i\}$ have no adjacent vertices except themselves, let's define ${\xi}_i = 1$ in the case, when $d_i = 0$. Also note, that our vertex $i\in V\setminus \cN(i)$, thus ksi-centrality $\xi_i$ is always greater or equal 1. Since the maximum number of edges from the neighborhood $\cN(i)$ to $V\setminus \cN(i)$ can be larger than $\Big|E\big(\cN(i), V\setminus \cN(i)\big)\Big|$ let's give

\begin{defin}
    For each vertex $i$ \bf{normalized ksi-centrality} $\hat\xi_{i}$ is defined by following
    $$  
        \hat{\xi}_i = \hat{\xi}(i) = \frac {\Big|E\big(\cN(i), V\setminus \cN(i)\big)\Big|} {\big|\cN(i)\big|\cdot\big|V\setminus \cN(i)\big|} = \frac {\Big|E\big(\cN(i), V\setminus \cN(i)\big)\Big|} {d_i (n-d_i)}.
    $$
\end{defin}

It is easy to see that by this definition $\frac 1 {n-d_i}\leq\hat{\xi}_i\leq 1$. Since $\frac {\Big|E\big(\cN(i), V\setminus \cN(i)\big)\Big|} {d_i(n-d_i)} = \frac {d_i} {d_i(n-d_i)} = \frac 1 {n-d_i}$, when vertices of $\cN(i)\cup \{i\}$ have no adjacent vertices except themselves, let's define $\hat{\xi}_i = \frac 1 n$ in the case, when $d_i = 0$.

Let's define for the whole graph $G$ average normalized ksi-coefficient.
\begin{defin}
    The average normalized ksi-coefficient
    $$  
        \hat\Xi(G) = \frac 1 n \sumt_{i\in V(G)} \hat\xi_i.
    $$
\end{defin}

\section{Calculations}

The author's colleague Ivan Samoylenko provided a fast algorithm for calculating ksi and normalized ksi-centralities based on the formula in the corollary~\ref{cor1}. This code is available on Github (\url{https://github.com/Samoylo57/ksi-centrality}).

\section{Results}
\subsection{Exponential law}
We computed ksi-centrality distributions for 40 different real-world networks (see figures~\ref{fig:1}--\ref{fig:5}) and more than 200 Erdos-Renyi graphs with different parameters. The real networks differ in types, properties, and relationships between nodes and edges. We provide a list of references in the table~\ref{tab:net}. We see that for each network from the table distribution of ksi-centrality centrality is similar to an exponential law and thus they are exponential-ksi. To show this more accurate, we constructed a linear approximation of the data in log y-scale and compared the resulting exponential function with distribution (see figure~\ref{fig:6}).

\begin{figure}[H]\vspace{-10pt}
    \centering
	\includegraphics[width = 0.71\textwidth]{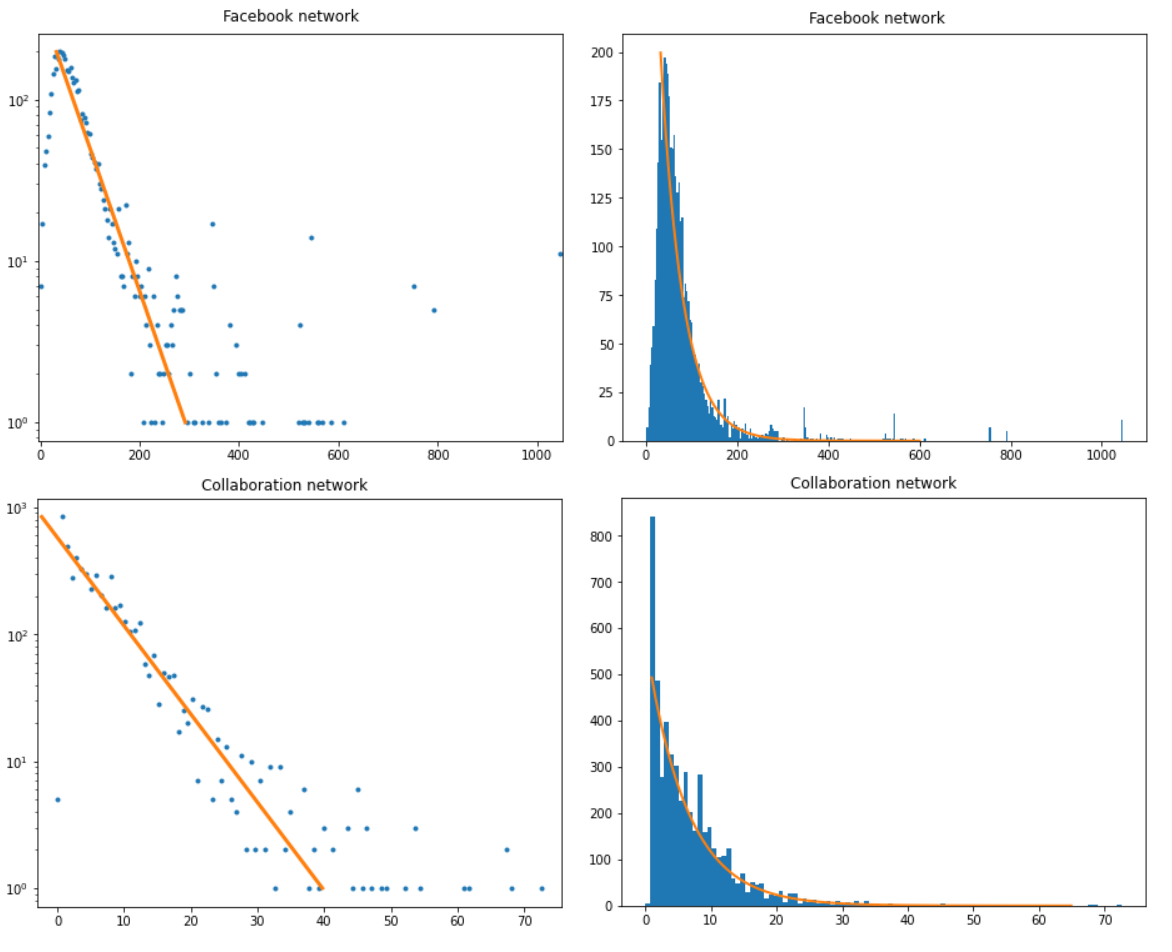}
	\caption{Linear approximation of ksi distribution on log-scale $y$ and comparison of fitted exponential function with it for Facebook and Collaboration networks.}
	\label{fig:6}
\end{figure}

The ksi distribution for each random network is bell-shaped (see figure~\ref{fig:12}).

\begin{figure}[H]
    \centering
	\includegraphics[width = 1\textwidth]{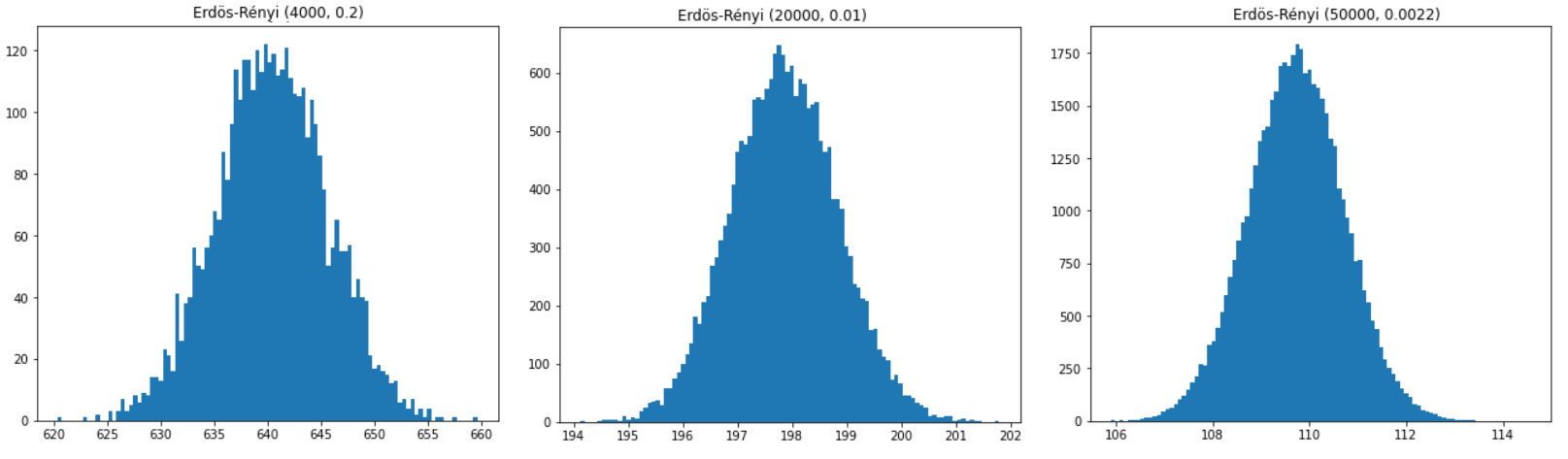}
	\caption{Ksi distribution for Erdos-Renyi graphs.}
	\label{fig:12}
\end{figure}

\subsection{Normalized ksi coefficient}
In the figure~\ref{fig:0} we have calculated the average normalized ksi coefficient for each of these 40 networks. Although the average normalized coefficient may tend to 0, for example for the Watts-Strogatz network, we see that for real networks with a large number of nodes the order of this value differs from $\frac 1 n$. For example, for the Stack Overflow network with 670816 nodes the average normalized ksi coefficient is $4\cdot10^{-4}$, while $\frac 1 {670816}\sim 10^{-6}$.

\subsection{Comparison with Watts-Strogatz and Barabasi-Albert networks}

In the figure~\ref{fig:7} we calculated the ksi-distributions for Barabasi-Albert network with 2000 nodes with respect to the preferential attachment parameter $m$. We see that for $10\leq m\leq 70$ the ksi-distributions in for Barabasi-Albert networks are similar to those of random graphs, they become left-skewed with larger $m$. Therefore, the Barabasi-Albert network for a small $m\leq 10$  is exponential-ksi.

For Watts-Strogatz networks the picture is similar. First, for each fixed parameters $(k, p)$ of the 2000-node Watts-Strogatz network, we fit the distribution of ksi to an exponential distribution, as in figure~\ref{fig:6}.
For each fit, we calculated the root mean squared difference between the fitted log distributions and the logarithm of original one. In the figure~\ref{fig:11} on the left and middle we see that the smallest deviation from the exponential distribution is observed for the smallest parameters $k$ and $p$, similar to the Barabasi-Albert case. In the figure~\ref{fig:11} on the right we see that for very small parameters $k$ and $p$ the distribution of ksi also differs from exponential. We also calculated the dependence of ksi-distributions to $k$ and $p$ separately. In the figure~\ref{fig:8} we see that for fixed $p = 0.016$ the distribution of ksi becomes the distribution for random networks with $k\geq 160$ and for $k\leq 30$ also differ from the distributions in real data. The same picture with fixed $k = 67$ with increasing $p$ (see figure~\ref{fig:9}).

\begin{figure}[H]
    \centering
	\includegraphics[width = 1.0\textwidth]{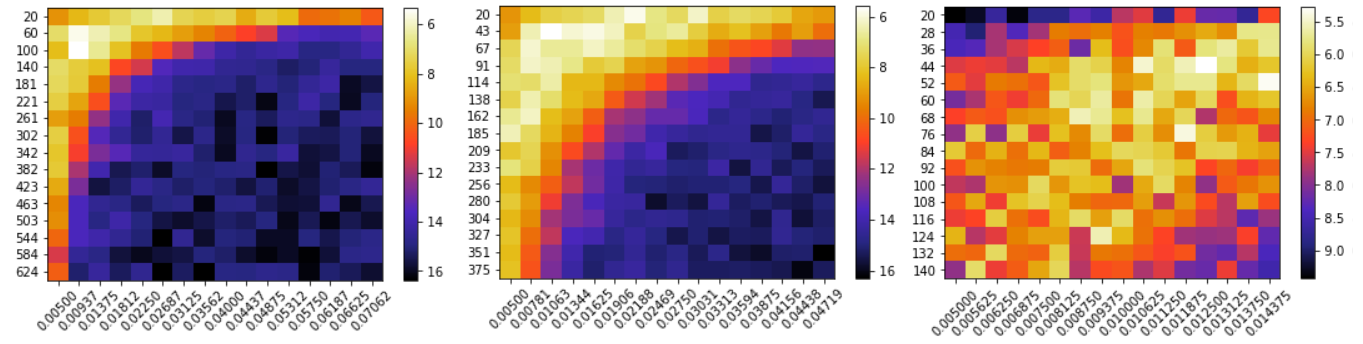}
	\caption{The root mean squared difference between the fitted (exponential) distributions and the original for fixed parameters $(k, p)$ of the 2000-node Watts-Strogatz network. The X-axis corresponds to $k$, and the Y-axis corresponds $p$.}
	\label{fig:11}
\end{figure}

We calculated ksi distributions with distributions similar to the exponential distribution for the Barabasi-Albert and Watts-Strogatz networks, also on a logarithmic scale (figure~\ref{fig:10}).

\begin{figure}[H]
    \centering
	\includegraphics[width = 0.8\textwidth]{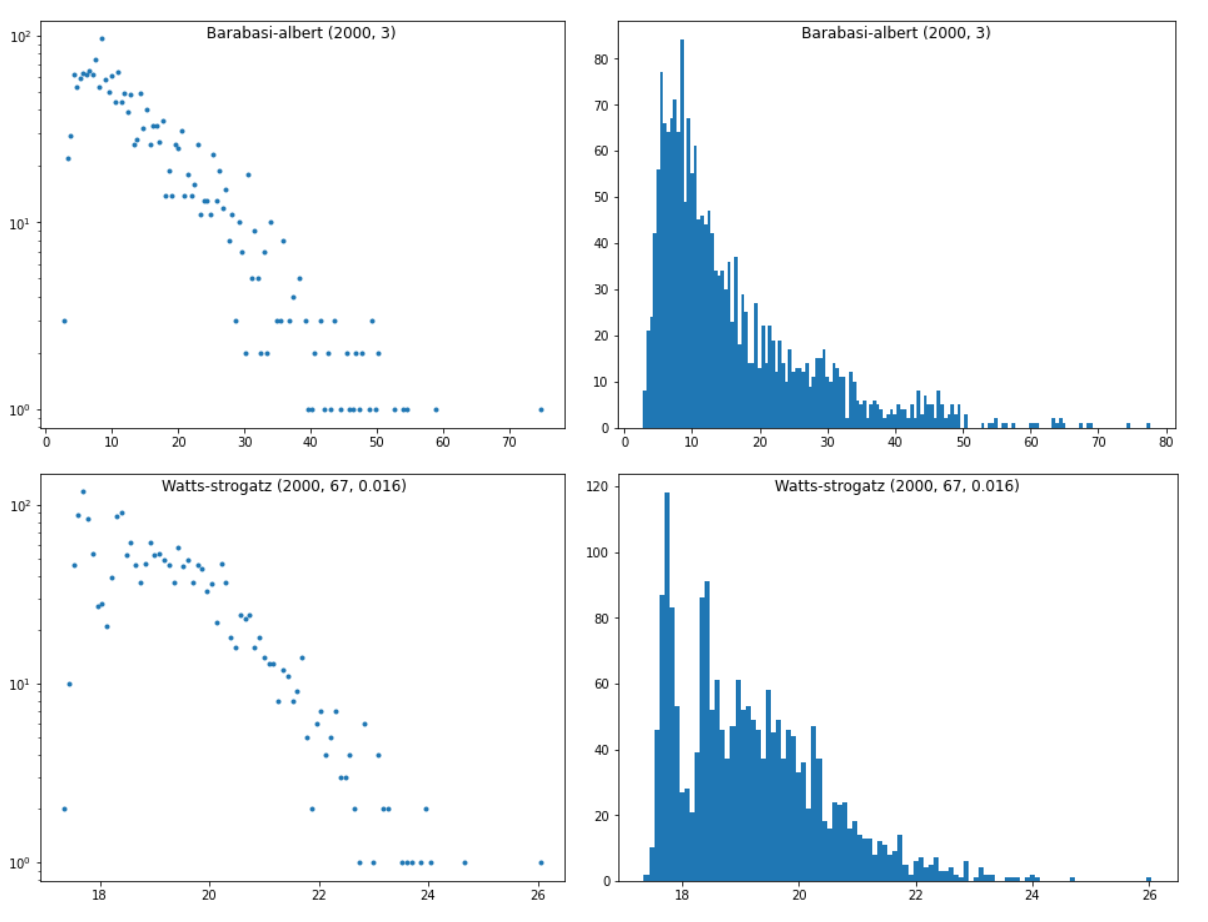}
	\caption{Initial ksi distributions (right) and in logarithmic scale (left) for the case where Barabasi-Albert and Watts-Strogatz networks are exponential-ksi networks. The number of nodes is 2000.}
	\label{fig:10}
\end{figure}

\section{Discussion}

We have shown that the distributions of ksi satisfy an exponential law for real networks while the distributions of ksi for random networks are bell-shaped and closer to the normal distribution. It turns out that the ksi distributions for Barabasi-Albert and Watts-Strogatz networks are similar to the ksi distributions for random networks (bell-shaped) for most parameters, but when these parameters become small enough, the Barabasi-Albert and Watts-Strogatz networks become more realistic (exponential-ksi) with respect to the ksi distributions. However, when the number of edges and the probability in the Watts-Strogatz network become very small, the gap between the number of vertices with large ksi and with small ksi becomes too large, and thus a large deviation from the exponential distribution appears. 

One explanation for this phenomenon may be that real networks have the so-called ``bow-tie'' structure, and the distribution of ksi-centrality exhibits exponential behavior, while random networks, as well as Barabasi-Albert and Watts-Strogatz networks with large parameters, are not bow-tie structures. If we consider sufficiently small parameters for the Barabasi-Albert and Watts-Strogatz networks, these networks become more tree-like and thus closer to the bow-tie structure. However, for the Watts-Strogatz network, the smaller the initial degrees, the less clustering in that network and the smaller the probability, the closer that network is to a ring lattice. For a ring lattice, the neighbor structure for each node is the same, and thus the ksi distribution is just a single number. Thus, there is a lower bound for the Watts-Strogatz network to be more ``realistic''.

\section{Acknowledgments}
We thank Ivan Samoylenko for useful comments and development of the code for calculation of ksi-centrality.

\newpage

\begin{figure}[h!]
    \centering
	\includegraphics[height = 0.9\textheight, width = 1\textwidth]{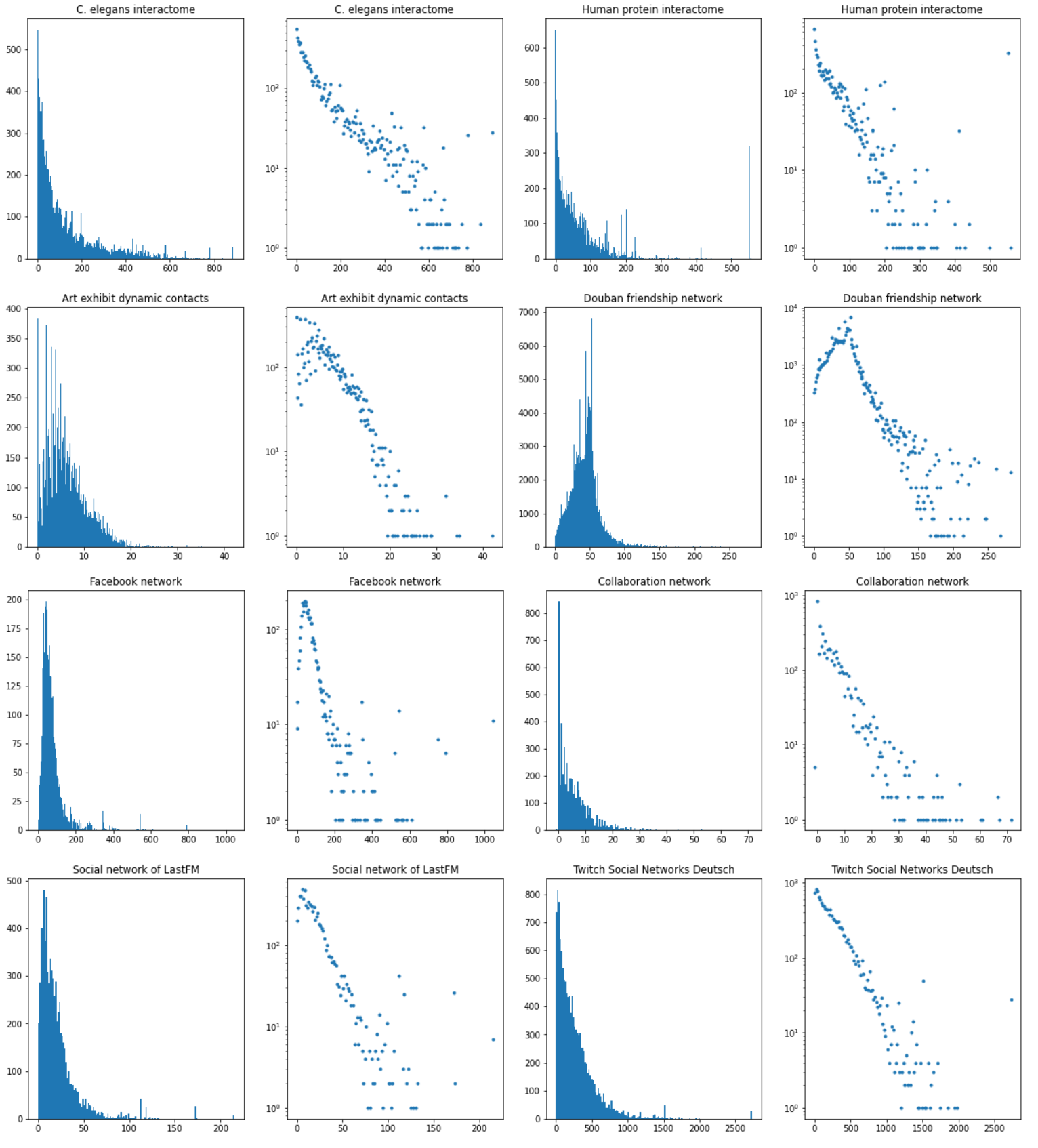}
	\caption{Ksi-distributions (left) and corresponding log-scale for y axes (right) for networks from table~\ref{tab:net}. The X-axis corresponds to the ksi value, and the Y-axis corresponds to the number of vertices with that ksi value.}
	\label{fig:1}
\end{figure}

\begin{figure}[h!]
    \centering
	\includegraphics[height = 0.9\textheight, width = 1\textwidth]{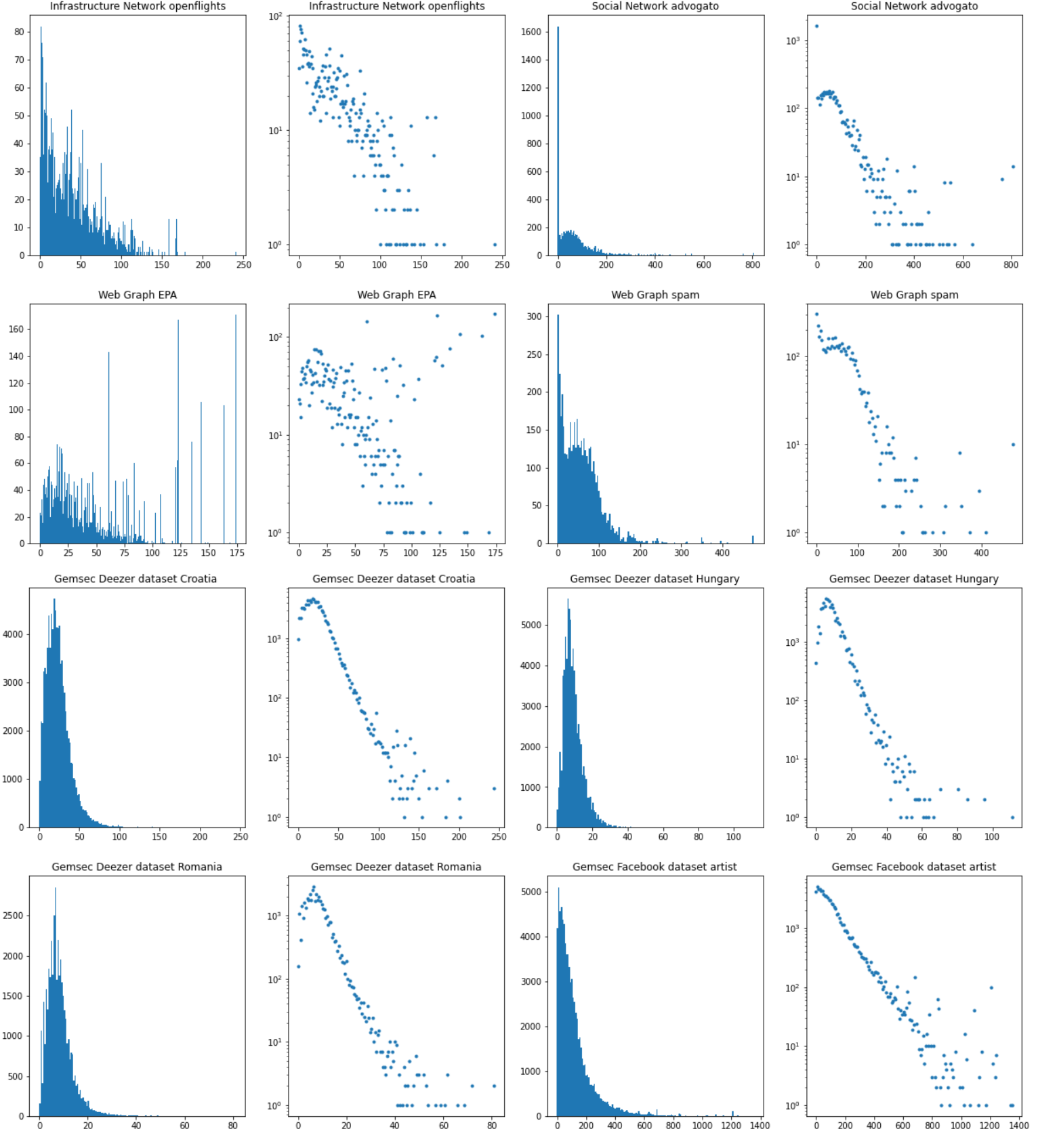}
	\caption{Ksi-distributions (left) and corresponding log-scale for y axes (right) for networks from table~\ref{tab:net}. The X-axis corresponds to the ksi value, and the Y-axis corresponds to the number of vertices with that ksi value.}
	\label{fig:2}
\end{figure}

\begin{figure}[h!]
    \centering
	\includegraphics[height = 0.9\textheight, width = 1\textwidth]{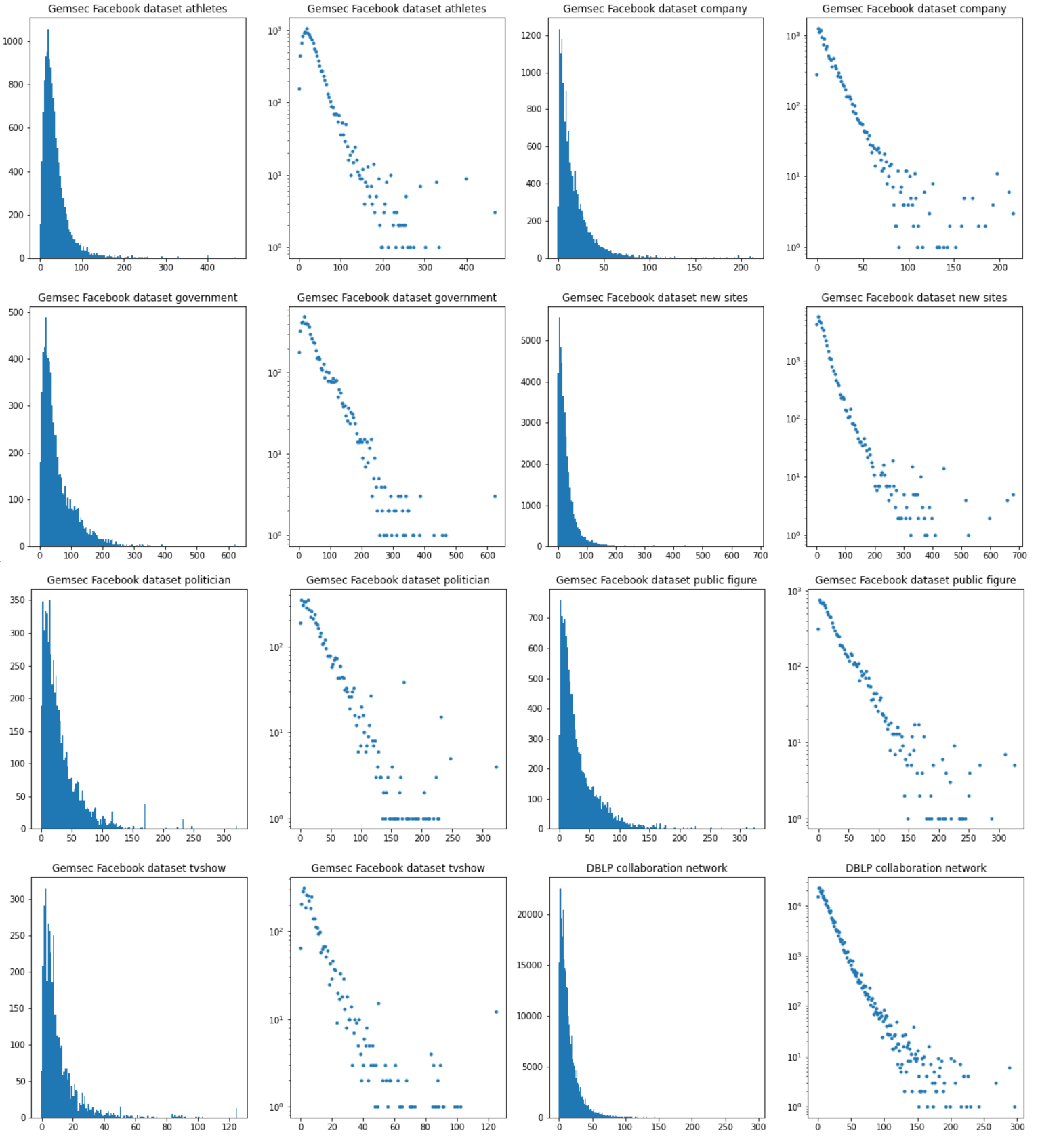}
	\caption{Ksi-distributions (left) and corresponding log-scale for y axes (right) for networks from table~\ref{tab:net}. The X-axis corresponds to the ksi value, and the Y-axis corresponds to the number of vertices with that ksi value.}
	\label{fig:3}
\end{figure}

\begin{figure}[h!]
    \centering
	\includegraphics[height = 0.9\textheight, width = 1\textwidth]{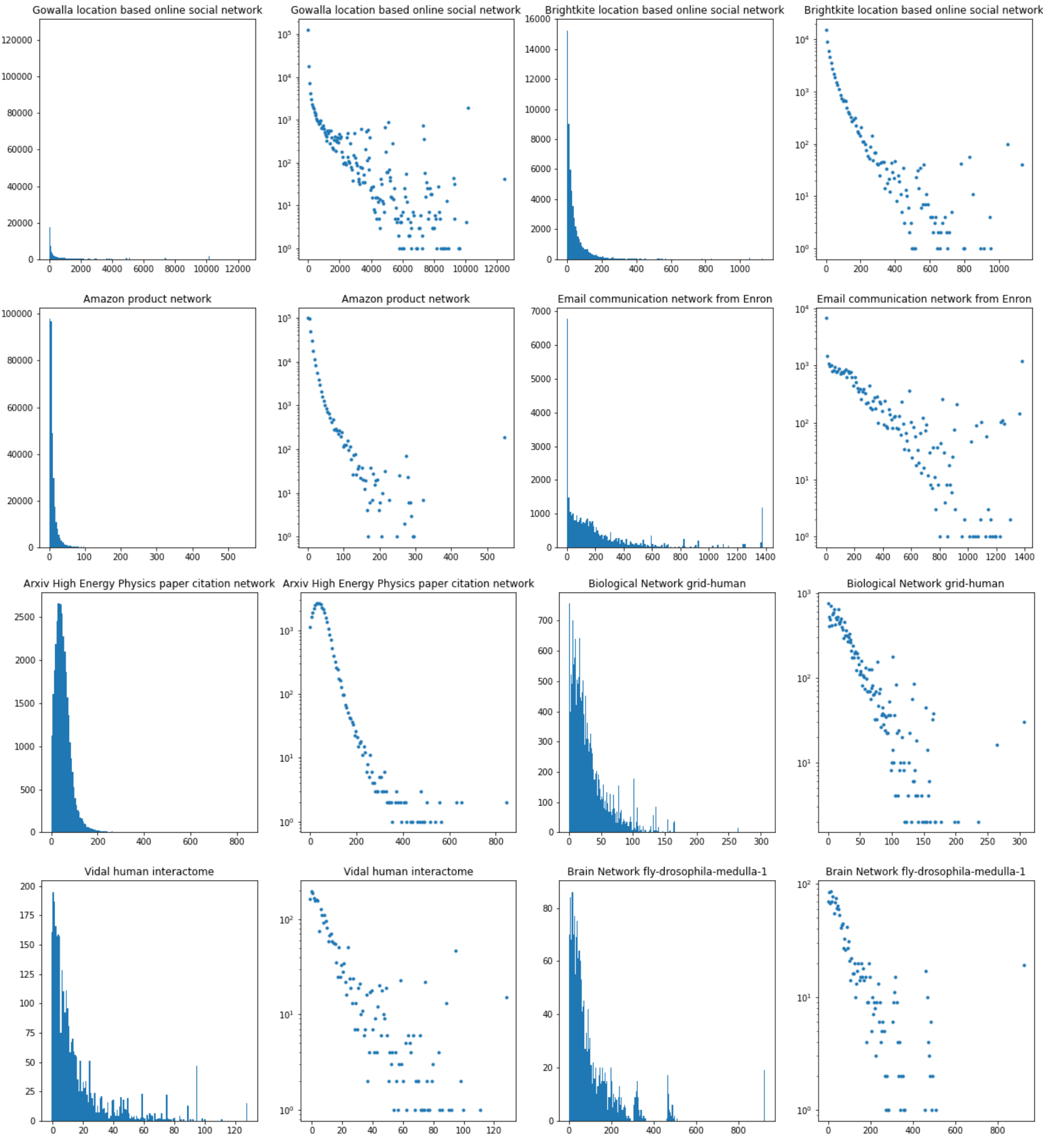}
	\caption{Ksi-distributions (left) and corresponding log-scale for y axes (right) for networks from table~\ref{tab:net}. The X-axis corresponds to the ksi value, and the Y-axis corresponds to the number of vertices with that ksi value.}
	\label{fig:4}
\end{figure}

\begin{figure}[h!]
    \centering
	\includegraphics[height = 0.9\textheight, width = 1\textwidth]{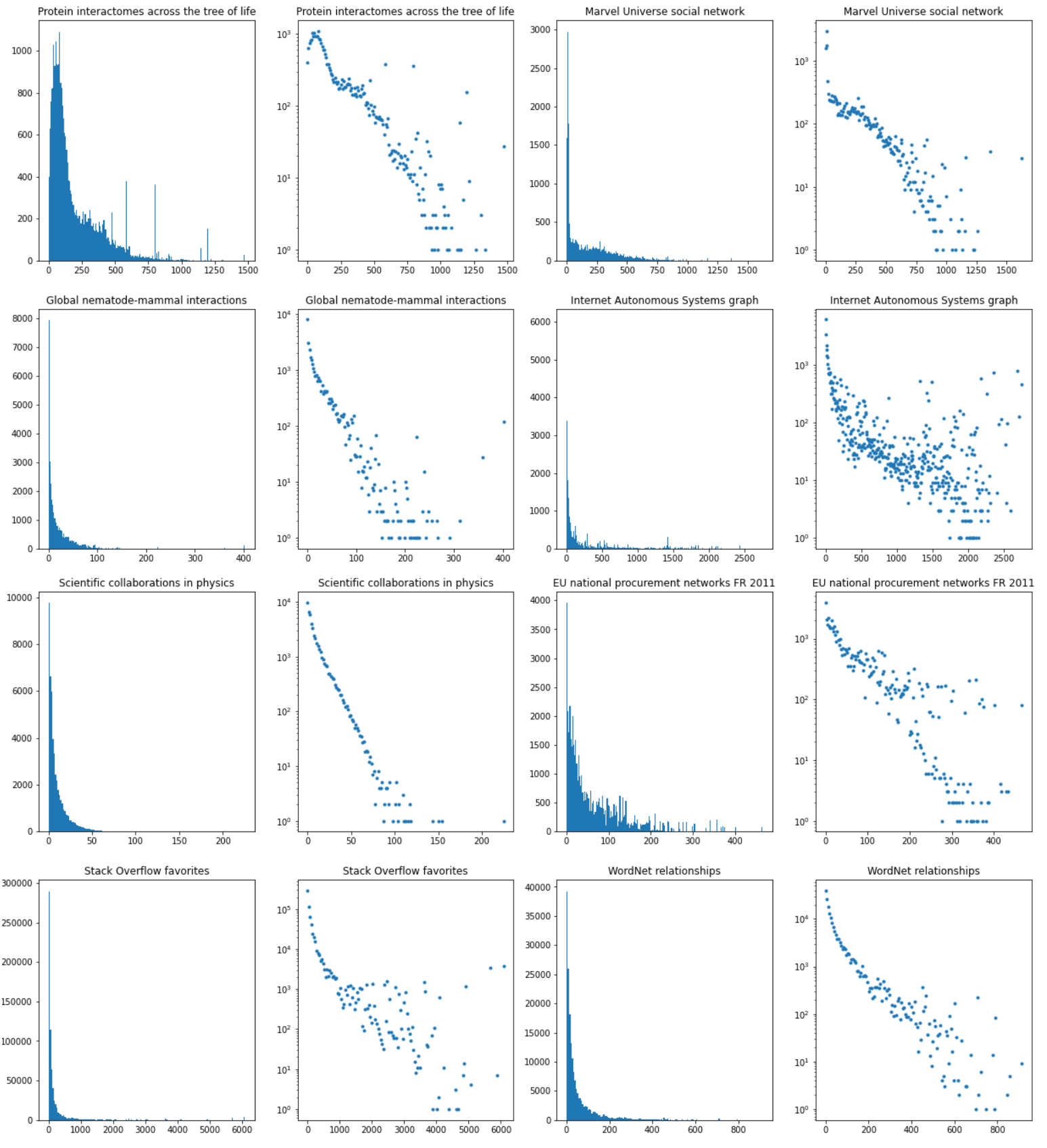}
	\caption{Ksi-distributions (left) and corresponding log-scale for y axes (right) for networks from table~\ref{tab:net}. The X-axis corresponds to the ksi value, and the Y-axis corresponds to the number of vertices with that ksi value.}
	\label{fig:5}
\end{figure}

\begin{table}[h!]
\resizebox{\textwidth}{!}{
\begin{tabular}{|c|c|}
     \hline
     Network name & Internet link \\
     \hline
    C. elegans interactome	\cite{bib1} &	\url{https://networks.skewed.de/net/celegans_interactomes} \\
Human protein interactome	\cite{bib2}&	\url{https://networks.skewed.de/net/reactome}	\\
Art exhibit dynamic contacts \cite{bib3}	&	\url{https://networks.skewed.de/net/sp_infectious}	\\
Douban friendship network	\cite{bib4}&	\url{https://networks.skewed.de/net/douban}	\\
Facebook network	\cite{bib5}&	\url{https://snap.stanford.edu/data/ego-Facebook.html}	\\
Collaboration network	\cite{bib6}&	\url{https://snap.stanford.edu/data/ca-GrQc.html}	\\
Social network of LastFM	\cite{bib7}&	\url{https://snap.stanford.edu/data/feather-lastfm-social.html}	\\
Twitch Social Networks Deutsch	\cite{bib8}&	\url{https://snap.stanford.edu/data/twitch-social-networks.html}	\\
Infrastructure Network openflights	\cite{bib9}&	\url{https://networkrepository.com/inf-openflights.php}	\\
Social Network advogato	\cite{bib9}&	\url{https://networkrepository.com/soc-advogato.php}	\\
Web Graph EPA	\cite{bib9}&	\url{https://networkrepository.com/web-EPA.php}	\\
Web Graph spam	\cite{bib9}&	\url{https://networkrepository.com/web-spam.php}	\\
Gemsec Deezer dataset Croatia 	\cite{bib10}&	\url{https://snap.stanford.edu/data/gemsec-Deezer.html}	\\
Gemsec Deezer dataset Hungary	\cite{bib10}&	\url{https://snap.stanford.edu/data/gemsec-Deezer.html}	\\
Gemsec Deezer dataset Romania	\cite{bib10}&	\url{https://snap.stanford.edu/data/gemsec-Deezer.html}	\\
Gemsec Facebook dataset artist	\cite{bib10}&	\url{https://snap.stanford.edu/data/gemsec-Facebook.html}	\\
Gemsec Facebook dataset athletes	\cite{bib10}&	\url{https://snap.stanford.edu/data/gemsec-Facebook.html}	\\
Gemsec Facebook dataset company	\cite{bib10}&	\url{https://snap.stanford.edu/data/gemsec-Facebook.html}	\\
Gemsec Facebook dataset government	\cite{bib10}&	\url{https://snap.stanford.edu/data/gemsec-Facebook.html}	\\
Gemsec Facebook dataset new sites	\cite{bib10}&	\url{https://snap.stanford.edu/data/gemsec-Facebook.html}	\\
Gemsec Facebook dataset politician	\cite{bib10}&	\url{https://snap.stanford.edu/data/gemsec-Facebook.html}	\\
Gemsec Facebook dataset public figure	\cite{bib10}&	\url{https://snap.stanford.edu/data/gemsec-Facebook.html	}\\
Gemsec Facebook dataset tvshow	\cite{bib10}&	\url{https://snap.stanford.edu/data/gemsec-Facebook.html}	\\
DBLP collaboration network	\cite{bib11}&	\url{https://snap.stanford.edu/data/com-DBLP.html}	\\
Gowalla location based online social	\cite{bib12}&	\url{https://snap.stanford.edu/data/loc-Gowalla.html}	\\
Brightkite location based online social	\cite{bib12}&	\url{https://snap.stanford.edu/data/loc-Brightkite.html}	\\
Amazon product network	\cite{bib11}&	\url{https://snap.stanford.edu/data/com-Amazon.html} \\
Email communication from Enron	\cite{bib13}&	\url{https://snap.stanford.edu/data/email-Enron.html}	\\
Arxiv High Energy paper citation	\cite{bib14}&	\url{https://snap.stanford.edu/data/cit-HepPh.html}	\\
Biological Network grid-human	\cite{bib9}&	\url{https://networkrepository.com/bio-grid-human.php}	\\
Vidal human interactome	\cite{bib15}&	\url{https://networks.skewed.de/net/interactome_vidal}	\\
Brain Network fly-drosophila-medulla	\cite{bib9}&	\url{https://networkrepository.com/bn-fly-drosophila-medulla-1.php}	\\
Protein interactomes across tree of life	\cite{bib16}&	\url{https://networks.skewed.de/net/tree-of-life}	\\
Marvel Universe social network	\cite{bib17}&	\url{https://networks.skewed.de/net/marvel_universe}	\\
Global nematode-mammal interactions	\cite{bib18}&	\url{https://networks.skewed.de/net/nematode_mammal}	\\
Internet Autonomous Systems graph	\cite{bib19}&	\url{https://networks.skewed.de/net/topology} \\
Scientific collaborations in physics	\cite{bib20}&	\url{https://networks.skewed.de/net/arxiv_collab}	\\
EU national procurement FR 2011	\cite{bib21}&	\url{https://networks.skewed.de/net/eu_procurements_alt}	\\
Stack Overflow favorites	\cite{bib22}&	\url{https://networks.skewed.de/net/stackoverflow}	\\
WordNet relationships	\cite{bib23}&	\url{https://networks.skewed.de/net/wordnet}	\\
\hline
\end{tabular}}

\caption{List of used networks with references.}
        \label{tab:net}
    \end{table}

\begin{figure}[h!]
    \centering
	\includegraphics[height = 0.9\textheight]{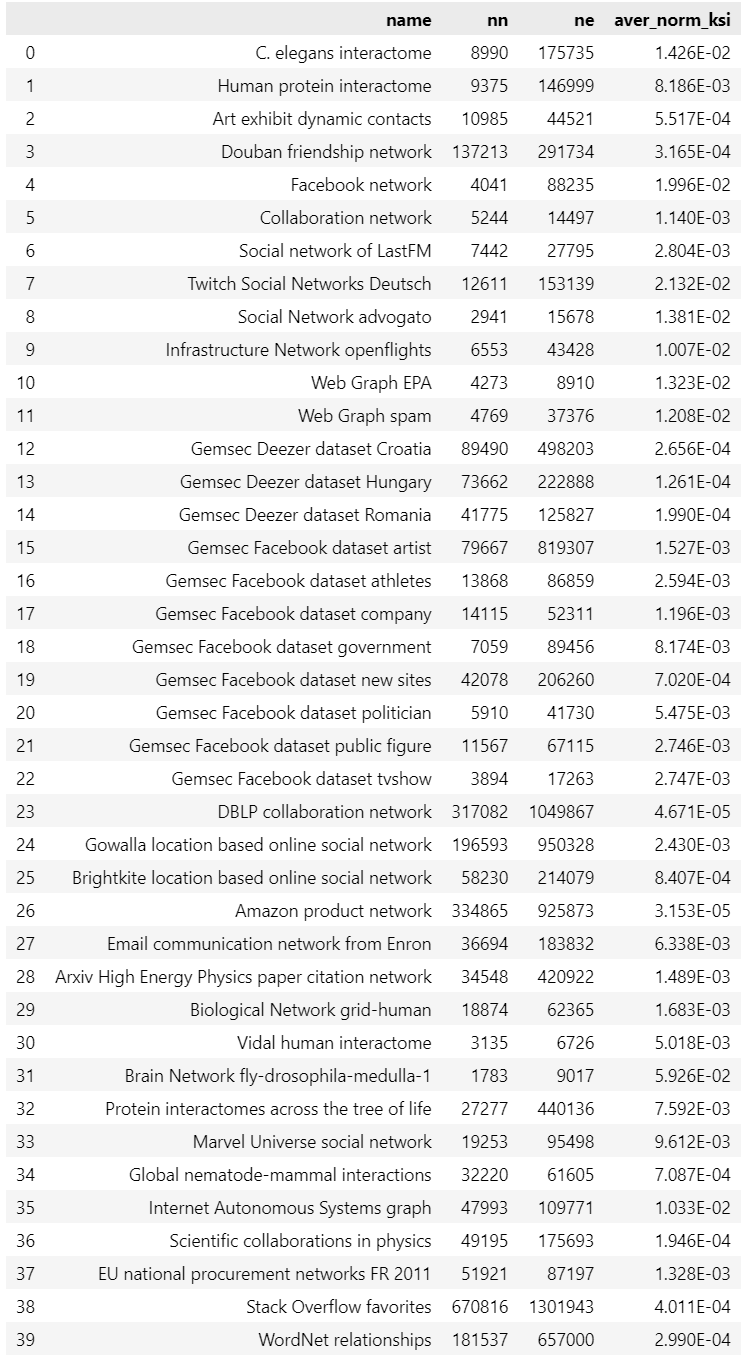}
	\caption{The average normalized ksi coefficient for each network from the table~\ref{tab:net}. ``nn'' is the number of nodes and ``ne'' is the number of edges for the corresponding network.}
	\label{fig:0}
\end{figure}

\begin{figure}[h!]
    \centering
	\includegraphics[height = 0.9\textheight, width = 1\textwidth]{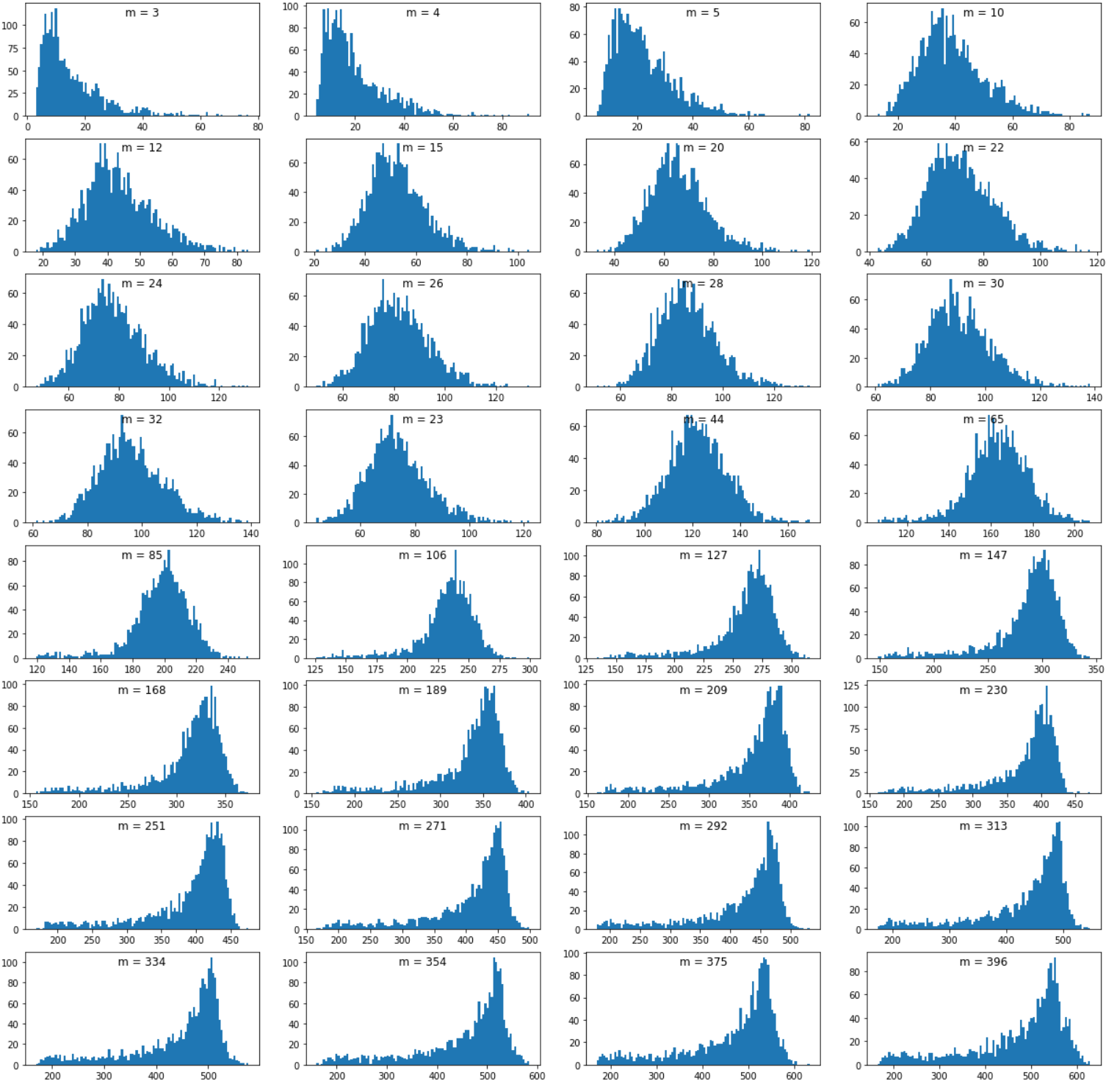}
	\caption{Dependence of the ksi-distribution on the number of preferential attachment parameter for the Barabasi-Albert network with 2000 nodes. The X-axis corresponds to the ksi value, and the Y-axis corresponds to the number of vertices with that ksi value.}
	\label{fig:7}
\end{figure}

\begin{figure}[h!]
    \centering
	\includegraphics[height = 0.9\textheight, width = 1\textwidth]{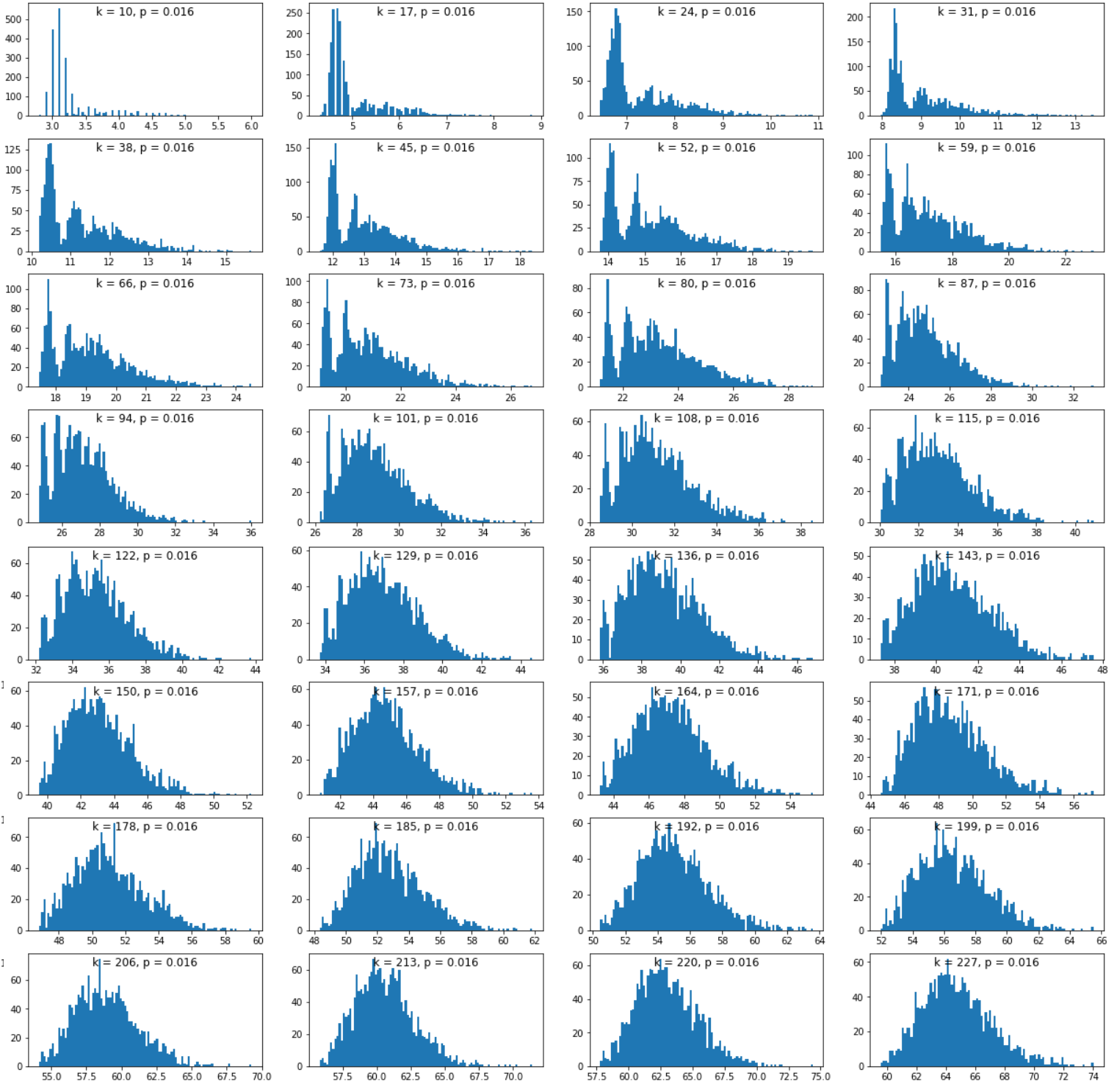}
	\caption{Dependence of the ksi-distribution on the number of initial vertices degree for the Watts-Strogatz network with 2000 nodes and $p = 0.016$. The X-axis corresponds to the ksi value, and the Y-axis corresponds to the number of vertices with that ksi value.}
	\label{fig:8}
\end{figure}

\begin{figure}[h!]
    \centering
	\includegraphics[height = 0.9\textheight, width = 1\textwidth]{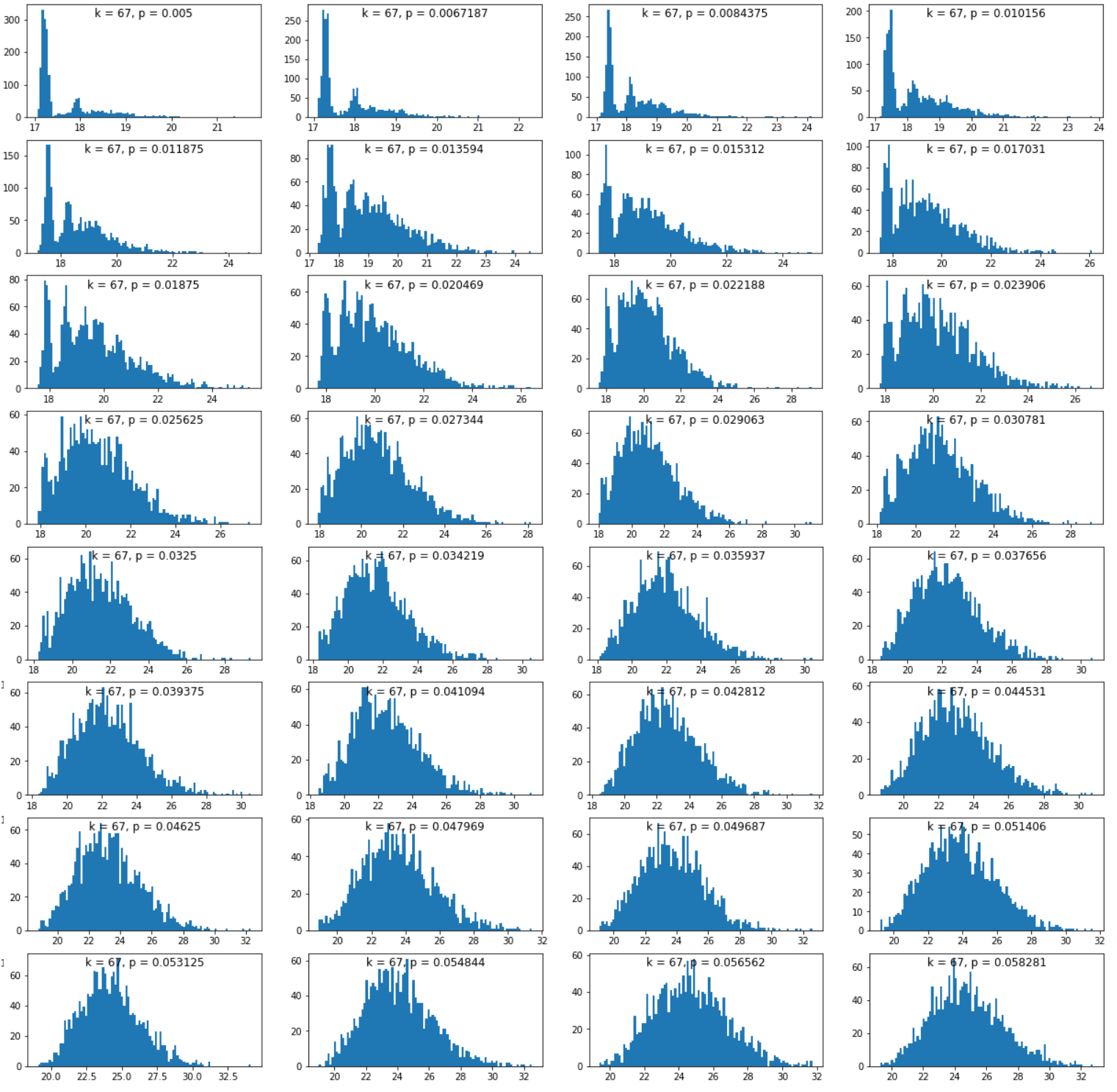}
	\caption{Dependence of the ksi-distribution on the probability $p$ for the Watts-Strogatz network with 2000 nodes and $k = 67$.  The X-axis corresponds to the ksi value, and the Y-axis corresponds to the number of vertices with that ksi value.}
	\label{fig:9}
\end{figure}


\begin{thebibliography}{99}

\bibitem{WS} 
Watts D. J., Strogatz S. H. Collective dynamics of ‘small-world’networks //nature. 1998. \bf{393}. №~6684. 440--442.

\bibitem{Tuz}
Tuzhilin M. "New centrality measure: ksi-centrality." arXiv preprint arXiv:2503.02488 (2025).

\bibitem{BA1}
Albert, R., Jeong, H., Barabási, A. L. (1999). Diameter of the world-wide web. nature, 401(6749), 130-131.

\bibitem{BA2}
Barabási, A. L., Albert, R. (1999). Emergence of scaling in random networks. science, 286(5439), 509-512.

\bibitem{SN1}
J. McAuley and J. Leskovec. Learning to Discover Social Circles in Ego Networks. NIPS, 2012. 

\bibitem{SN2}
J. Leskovec, J. Kleinberg and C. Faloutsos. Graph Evolution: Densification and Shrinking Diameters. ACM Transactions on Knowledge Discovery from Data (ACM TKDD), 1(1), 2007.

\bibitem{SN3}
B. Rozemberczki and R. Sarkar. Characteristic Functions on Graphs: Birds of a Feather, from Statistical Descriptors to Parametric Models. 2020.

\bibitem{SN4}
Varshney, L. R., Chen, B. L., Paniagua, E., Hall, D. H., and Chklovskii, D. B. (2011). Structural properties of the Caenorhabditis elegans neuronal network. PLoS computational biology, 7(2), e1001066.

\bibitem{Free}
Freeman, L. C. (2002). Centrality in social networks: Conceptual clarification. Social network: critical concepts in sociology. Londres: Routledge, 1(3), 238-263.

\bibitem{bib1}
N. Simonis et al., "Empirically controlled mapping of the Caenorhabditis elegans protein-protein interactome network." Nature Methods 6(1), 47-54 (2009).

\bibitem{bib2}
G. Joshi-Tope et al. "Reactome: a knowledgebase of biological pathways." Nucleic Acids Research 33.suppl 1, D428-D432 (2005).

\bibitem{bib3}
L. Isella et al., "What's in a crowd? Analysis of face-to-face behavioral networks." Journal of Theoretical Biology 271, 166-180 (2011).

\bibitem{bib4}
J. Kunegis, "Douban network dataset." KONECT, the Koblenz Network Collection (2016).

\bibitem{bib5}
J. McAuley and J. Leskovec. Learning to Discover Social Circles in Ego Networks. NIPS, 2012.

\bibitem{bib6}
J. Leskovec, J. Kleinberg and C. Faloutsos. Graph Evolution: Densification and Shrinking Diameters. ACM Transactions on Knowledge Discovery from Data (ACM TKDD), 1(1), 2007.

\bibitem{bib7}
B. Rozemberczki and R. Sarkar. Characteristic Functions on Graphs: Birds of a Feather, from Statistical Descriptors to Parametric Models. 2020.

\bibitem{bib8}
B. Rozemberczki, C. Allen and R. Sarkar. Multi-scale Attributed Node Embedding. 2019.

\bibitem{bib9}
Rossi, R., \& Ahmed, N. (2015, March). The network data repository with interactive graph analytics and visualization. In Proceedings of the AAAI conference on artificial intelligence (Vol. 29, No. 1).

\bibitem{bib10}
B. Rozemberczki, R. Davies, R. Sarkar and C. Sutton. GEMSEC: Graph Embedding with Self Clustering. 2018. 

\bibitem{bib105}
B. Rozemberczki, C. Allen and R. Sarkar. Multi-scale Attributed Node Embedding. 2019.

\bibitem{bib11}
J. Yang and J. Leskovec. Defining and Evaluating Network Communities based on Ground-truth. ICDM, 2012.

\bibitem{bib12}
E. Cho, S. A. Myers, J. Leskovec. Friendship and Mobility: Friendship and Mobility: User Movement in Location-Based Social Networks ACM SIGKDD International Conference on Knowledge Discovery and Data Mining (KDD), 2011.

\bibitem{bib13}
J. Leskovec, K. Lang, A. Dasgupta, M. Mahoney. Community Structure in Large Networks: Natural Cluster Sizes and the Absence of Large Well-Defined Clusters. Internet Mathematics 6(1) 29--123, 2009.

\bibitem{bib14}
J. Leskovec, J. Kleinberg and C. Faloutsos. Graphs over Time: Densification Laws, Shrinking Diameters and Possible Explanations. ACM SIGKDD International Conference on Knowledge Discovery and Data Mining (KDD), 2005.

\bibitem{bib15}
JF. Rual et al. "Towards a proteome-scale map of the human protein-protein interaction network." Nature 437.7062, 1173-1178. (2005).

\bibitem{bib16}
Marinka Zitnik, Rok Sosic, Marcus W Feldman, and Jure Leskovec. "Evolution of resilience in protein interactomes across the tree of life", Proceedings of the National Academy of Sciences (PNAS), 2019, 116 (10) 4426-4433.

\bibitem{bib17}
R. Alberich, J. Miro-Julia, and F. Rossello, "Marvel Universe looks almost like a real social network." arxiv:cond-mat/0202174 (2002).

\bibitem{bib18}
T. Dallas et al. "Gauging support for macroecological patterns in helminth parasites." Glob. Ecol. Biogeogr. 27, 1437-1447 (2018).

\bibitem{bib19}
B. Zhang et al., "Collecting the Internet AS-level topology." SIGCOMM Computer Communication Review 35(1), 53-61 (2005).

\bibitem{bib20}
M. E. J. Newman, "The structure of scientific collaboration networks." Proc. Natl. Acad. Sci. USA 98(2), 404-409 (2001).

\bibitem{bib21}
J. Wachs, M. Fazekas, and J. Kertesz, "Corruption Risk in Contracting Markets: A Network Science Perspective." International Journal of Data Science and Analytics, pp 1-16 (2020).

\bibitem{bib22}
Kunegis, "Stack overflow network dataset" KONECT, the Koblenz Network Collection (2016).

\bibitem{bib23}
C. Fellbaum. "WordNet." Theory and Applications of Ontology: Computer Applications. Springer, 231-243 (2010).



\end{thebibliography}
\end{document}